\documentclass[12pt,dvipsnames]{amsart}

\usepackage[english]{babel}
\usepackage[utf8]{inputenc}
\usepackage{amssymb,amsthm,amsfonts,amsmath}
\usepackage{mathrsfs}
\usepackage{verbatim}
\usepackage{mathtools}
\usepackage[shortlabels]{enumitem}

\usepackage{tikz}
\usepackage{pgfmath}
\usetikzlibrary{matrix,arrows}
\usetikzlibrary {graphs.standard}
\usetikzlibrary{decorations.pathmorphing,positioning}
\tikzset{mydeco/.style={decoration={random steps,segment length=.08em,amplitude=.05em},decorate,line cap=round}}
\usepackage[dvipsnames]{xcolor}

\usepackage[normalem]{ulem}
\usepackage{fontawesome5}

\allowdisplaybreaks

	\definecolor{apricot}{rgb}{0.98, 0.81, 0.69}
\definecolor{aquamarine}{rgb}{0.5, 1.0, 0.83}

\usepackage{autobreak}

\usepackage{url}
\usepackage{amsopn}
\usepackage{graphicx}

\usepackage{enumitem,titletoc}
\definecolor{darkblue}{rgb}{0.0, 0.0, 0.55}
\usepackage[backref=page,colorlinks,breaklinks,linkcolor=BrickRed,citecolor=OliveGreen,urlcolor=darkblue,hypertexnames=true]{hyperref}
\usepackage[a4paper,margin=2.5cm]{geometry}
\usepackage{ytableau}

\usepackage[baseline]{euflag}

\usepackage{setspace}
\makeatletter
\newtheorem*{rep@thm}{\rep@title}
\newcommand{\newreptheorem}[2]{%
	\newenvironment{rep#1}[1]{%
		\def\rep@title{#2 \ref{##1}}%
		\begin{rep@thm}}%
		{\end{rep@thm}}}
\makeatother
\newreptheorem{thm}{Theorem}
\newreptheorem{prop}{Proposition}
\newreptheorem{cor}{Corollary}

\def\la{\lambda}

\DeclareMathOperator{\EE}{E}

\DeclareMathOperator{\het}{ht}

\def\la{\lambda}

\numberwithin{equation}{section}

\newcommand{\N}{\mathbb N}

\newcommand{\CC}{\mathbb C}

\newtheorem{theorem}{Theorem}[section]

\newtheorem{lemma}[theorem]{Lemma}
\newtheorem{proposition}[theorem]{Proposition}

\theoremstyle{definition}
\newtheorem{definition}[theorem]{Definition}
\newtheorem{remark}[theorem]{Remark}
\newtheorem{example}[theorem]{Example}

\begin{document}

\title{Quantum Max Cut for complete tripartite graphs}

\author[T. \v{S}trekelj]{Tea \v{S}trekelj}
\address{Department of Mathematics, FAMNIT, University of Primorska, Koper \&  Institute of Mathematics, Physics and Mechanics, Ljubljana, Slovenia.}
\email{tea.strekelj@famnit.upr.si}

\thanks{The author was supported by the Slovenian Research and Innovation Agency grant J1-60011.}



\begin{abstract}
The Quantum Max-$d$-Cut ($d$-QMC) problem is a special instance of a $2$-local Hamiltonian problem, representing the quantum analog of the classical Max-$d$-Cut problem. The $d$-QMC problem seeks the largest eigenvalue of a Hamiltonian defined on a graph with $n$ vertices, where edges correspond to swap operators acting on $(\mathbb{C}^d)^{\otimes n}$. In recent years, progress has been made by investigating the algebraic structure of the $d$-QMC Hamiltonian. Building on this approach, this article solves the $d$-QMC problem for complete tripartite graphs for small local dimensions, $d \le 3$.

\end{abstract}

\maketitle




\newpage 

\section{Introduction}

The $k$-local Hamiltonian problem is a central problem in quantum complexity theory. The goal is to determine the smallest eigenvalue (or ground state energy) of a Hamiltonian $H$, which is a self-adjoint matrix acting on the space of $n$ qubits, $(\mathbb{C}^2)^{\otimes n}$. This matrix, of dimension $2^n \times 2^n$, is structured as a sum of local terms. In fact, for a chosen $k\leq n,$
$$
H = \sum_{\substack{S \subseteq \{1,\ldots,n\} \\[.5mm] |S| = k}} H_S,
$$
where each $H_S$ acts nontrivially on a subset (of $S$) of at most $k$ qubits and as the identity on the rest. Such an $H$ is called a $k$-local Hamiltonian.

We investigate generalizations of the Quantum Max Cut (QMC) problem, a special instance of the $2$-local Hamiltonian problem. Named by Gharibian and Parekh~\cite{GP19}, QMC serves as a quantum analogue of the classical Max Cut (MC) problem, with its corresponding input matrix known as the QMC Hamiltonian. The problem arises naturally in physics, where it corresponds to finding the ground state energy of the anti-ferromagnetic Heisenberg model. This model is crucial for describing the magnetic properties of insular crystals under the $2$-local assumption that only nearest-neighbour interactions are significant~\cite{Aue,BDZ}.

Beyond its physics origins, the QMC problem has garnered significant attention in computational complexity theory. As a simple prototype of a QMA-complete problem~\cite{PM17}, it has become a valuable testbed for designing approximation algorithms for other QMA-hard problems~\cite{AMG20,PT21,PT22,Lee22,Kin23}. Progress on solving the QMC problem includes exact analytical solutions for specific graph families, such as complete bipartite graphs~\cite{LM62} and one-dimensional chains~\cite{LM16}. More recently, second-order cone relaxations that provide approximations for large graphs~\cite{HTPG} were developed.

In this article, we address the $d$-QMC problem, a generalization of the QMC problem to systems of $n$ qudits. The Hamiltonian for this problem acts on the space $(\mathbb{C}^d)^{\otimes n}$. The motivation for studying the $d$-QMC problem stems from the challenge of determining the ground state energies of $\operatorname{SU}(d)$-Heisenberg models on lattices~\cite{KT,BAMC,PM21}. Furthermore, it serves as the quantum analogue of both the classical Max-$d$-Cut problem, which concerns finding maximum $d$-colorable subgraphs, and the anti-ferromagnetic $d$-state Potts model~\cite{FJ}.

The focus of this paper is on the $d$-QMC problem pertaining to a specific family of graphs, namely complete tripartite graphs. We tackle the special case where the local dimension $d$ is small, i.e., $d\leq3.$

\subsection{Quantum Max $d$-Cut and swap matrices}
A QMC Hamiltonian pertains to a given graph $G$ on say $n$ vertices. Throughout, we denote the vertex set of $G$ by V$(G)$ and the edge set of $G$ by E$(G).$ 
To address the $d$-QMC problem, we investigate the algebraic structure of the QMC Hamiltonian, following the approach of  \cite{BCEHK24,TRZ,KSV+}.
This approach starts by expressing the QMC Hamiltonian in terms of the {swap matrices} Swap$_{ij}^{(d)}.$\footnote{In physics literature, these are often called SWAP or exchange operators \cite{NC}.
} 

\begin{definition} Fix $n$ and $1\leq i < j \leq n.$ The (\textbf{qudit}) \textbf{swap matrix} Swap$_{ij}^{(d)}$  is defined by its action on rank one tensors as
	$$
	\text{Swap}_{ij}^{(d)}(v_1 \otimes \cdots \otimes v_i \otimes \cdots \otimes v_j \otimes \cdots \otimes v_n) = v_1 \otimes \cdots \otimes v_j \otimes \cdots \otimes v_i \otimes \cdots \otimes v_n
	$$
 for any $v_1, \ldots, v_n \in \CC^d.$
\end{definition}

	The action of swap matrices on a qudit space defines a representation $\rho_n^{(d)}$ of the symmetric group $S_n$ on $(\mathbb{C}^d)^{\otimes n}$ defined by
	$$
	\rho_n^{(d)}(\pi) (v_1 \otimes \cdots \otimes v_n) = v_{\pi^{-1}(1)} \otimes \cdots \otimes v_{\pi^{-1}(n)}.
	$$
	Denoting by $\CC[S_n]$ the group algebra of $S_n,$
	the image $\rho_n^{(d)}\big(\CC[S_n] \big)$ is a subalgebra of $M_{d^n}(\CC).$ We denote it by $M^{\text{Sw}_d}_n(\mathbb{C})$ and call it the \textbf{$d$-swap algebra}. 

\begin{definition}\label{def:qmc}
     For a graph $G$ on $n$ vertices, the \textbf{Quantum Max $d$-Cut} \textbf{Hamiltonian} ($d$-QMC Hamiltonian) pertaining to $G$ is defined as
    \begin{equation}\label{eq:qmcd}\tag{$H_G^d$}
        H_G^{d} = \sum_{(i,j) \in \EE(G)} 2 \left(I - \textnormal{Swap}^{(d)}_{ij}\right).
    \end{equation}
\end{definition}
The \textbf{$d$-QMC problem} is about determining the largest eigenvalue of the $d$-QMC Hamiltonian $H_G^d$ in \eqref{eq:qmcd}. 
The original QMC problem, for local dimension~$d=2$, was formulated with a Hamiltonian built from Pauli matrices. 
This approach was later generalized to~$d>2$ by~\cite{CJKKW23+}, who employed Gell-Mann matrices as a natural higher-dimensional extension of the Pauli matrices. 
The alternative formulation using swap matrices as in Definition \ref{def:qmc} was first studied in~\cite{BCEHK24} for $d=2$ and subsequently  by~\cite{KSV+} for general~$d.$ 

The QMC (and $d$-QMC) problem, as an instance of the general $2$-local Hamiltonian problem, is known to be computationally intractable. It belongs to the Quantum Merlin Arthur (QMA)-hard complexity class~\cite{KSV02,KKR06}, the quantum analogue of NP-hard. Given that the QMC problem is hard to solve, research has focused on restricted cases, such as finding exact solutions for specific instances~\cite{LM62}, designing efficient high-precision approximation algorithms for the largest eigenvalue~\cite{LVV15}, or achieving constant-factor approximations~\cite{GK12,BH13,BGKT19,HM17}. Other work has explored the hardness of computing ground state properties~\cite{GH}. In this paper, we leverage the algebraic structure of swap matrices to solve the $d$-QMC problem for $d = 2,3$ on the family of complete tripartite graphs $K_{p,q,r}$ with $p \geq q \geq r.$

\begin{theorem}\label{th:prelim}
	Let $n, p,q,r \in \N$ be such that $p\geq q\geq r >0$ and $p+q+r=n.$
	\begin{enumerate}[(a)]
		\item The solution to the $2$-QMC problem for $K_{p,q,r}$ is
		\begin{enumerate}[(i)]
			\item $2(n-p)(p+1)$ if $p\geq q+r;$
			\item $4k^2-1$ if $n=2k$ and $p < q+r;$
			\item $4k(k+1)-3$ if $n=2k+1$ and $p< q+r.$
		\end{enumerate}
		
		\item The solution to the $3$-QMC problem for $K_{p,q,r}$ is  $2 n (2 + p + q) - 2 (p^2 + q + q^2 + p (2 + q)).$
	\end{enumerate}
\end{theorem}

While the above theorem captures the essence of the main result, the proper statement requires additional background. We dedicate the following two subsections to the requisite combinatorial theory.

\subsection{Partitions and irreducible representations of $S_n$}
The $d$-swap algebra $M^{\text{Sw}_d}_n(\mathbb{C})$ is semisimple as it is the image of a representation of $\CC[S_n].$ It hence decomposes into a direct sum of (simple) matrix algebras that are images of some of the \textit{irreducible representations (irreps)} of $S_n.$
It is well-known that the irreps of $S_n$ are indexed by partitions $\lambda$ of $n$ (denoted by $\lambda \vdash n$), where
$$
\lambda = (\lambda_1, \ldots, \lambda_k)\in\N^k, \quad \lambda_1 \geq \cdots \geq \lambda_k > 0, 
\quad \sum_{i=1}^k \lambda_i = n.
$$
The number $k$ is referred to as the \textbf{height} of $\lambda$ and denoted by $k=\het(\lambda)$. For any $\ell<k$, the partition $(\lambda_1,\ldots,\lambda_\ell)$ is called a \textbf{head}, and the partition $(\lambda_{l+1},\ldots,\lambda_f)$ is called a \textbf{tail} of $\lambda$. 
A partition $\lambda \vdash n$ is often represented by its \textbf{Young diagram}. A Young diagram of shape $\lambda = (\lambda_1, \ldots, \lambda_k)$ has $n$ boxes that are distributed into $k$ rows, where the $i$th row consists of $\lambda_i$ boxes. E.g., if $\lambda= (4,2,1) \vdash 7,$ then $\het(\lambda)=3$ and its Young diagram looks like
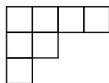
\[
\ytableausetup{smalltableaux}
\ydiagram{4,2,1}
\]

\vspace{2mm}

A \textbf{Young tableau} of shape $\la$ is a filling of the $n$
boxes of the corresponding Young diagram with the integers $1,\ldots,n$ such that each integer is used exactly once.

\subsection{Decomposition of the $d$-swap algebra}
The precise decomposition of $M^{\text{Sw}_d}_n(\mathbb{C})$ into irreps as presented in the following theorem can be derived by applying the Schur-Weyl duality. 

\begin{theorem}{\cite[Theorem 2.2]{KSV+}}
	The $d$-swap algebra $M^{\text{Sw}_d}_n(\mathbb{C})$ decomposes into a direct sum of simple algebras generated by the irreps $\rho_\lambda$ of $S_n$ corresponding to partitions of $n$ with at most $d$ rows,
	$$
	M^{\text{\textnormal{Sw}}_d}_n(\mathbb{C}) \cong \bigoplus_{\substack{\lambda\vdash n \\[.5mm] \het(\lambda)\le d}} \rho_{\lambda} (\mathbb{C}S_n).
	$$
\end{theorem}
The $d$-QMC Hamiltonian thus becomes block diagonal with an appropriate choice of a basis and when solving the $d$-QMC problem, the large-scale eigenvalue computation is broken down into several independent and more manageable calculations on each block.

In addition to the above decomposition, the precise relations defining the $d$-swap algebra are known (see \cite{BCEHK24} for $d=2$ and \cite{KSV+} for the general case). In fact, 
the $d$-swap algebra  is isomorphic to 
the quotient of the free algebra
$\CC\langle\text{swap}_{ij}\colon 1\le i<j\le n\rangle$ generated by the $\binom{n}{2}$ freely noncommuting variables modulo the relations defining $S_n$ plus one single relation. The latter is equivalent to the vanishing of an antisymmetrizer on $d+1$ letters on $(\CC^d)^{\otimes n}$ 
(\cite[Section 11.6]{Probook}, see also \cite[Theorem 2.8.1]{Gri}).

\section{Preliminaries}

We denote by $K_{p,q,r}$ the tripartite graph on $n=p+q+r$ vertices. This means that the $n$ vertices are divided into $3$ disjoint subsets of size $p,q$ and $r,$ respectively, and every edge in the graph connects a vertex in one set to a vertex in a different set. We will assume without loss of generality that $p \geq q \geq r>0.$ E.g., for $p=q=3$ and $r=2,$ the graph $K_{3,3,2}$ looks like
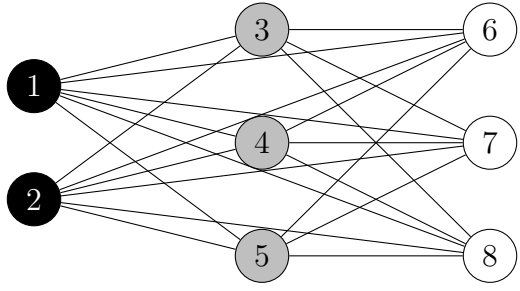
\begin{center}
\begin{tikzpicture}[
	scale=1, 
	every node/.style={transform shape}, 
	v1_style/.style={circle, draw, text=white, fill=black, minimum size=20pt, inner sep=0pt},
	v2_style/.style={circle, draw, fill=lightgray, minimum size=20pt, inner sep=0pt},
	v3_style/.style={circle, draw, fill=white, minimum size=20pt, inner sep=0pt},
	partition_label/.style={font=\Large\bfseries}
	]
	
	\def\xdist{3} 
	\def\ydist{1.5} 
	
	\node[v1_style] (v11) at (0, 0.5*\ydist) {$1$};
	\node[v1_style] (v12) at (0, -0.5*\ydist) {$2$};
	
	\node[v2_style] (v21) at (\xdist, 1*\ydist) {$3$};
	\node[v2_style] (v22) at (\xdist, 0) {$4$};
	\node[v2_style] (v23) at (\xdist, -1*\ydist) {$5$};
	
	\node[v3_style] (v31) at (2*\xdist, 1*\ydist) {$6$};
	\node[v3_style] (v32) at (2*\xdist, 0) {$7$};
	\node[v3_style] (v33) at (2*\xdist, -1*\ydist) {$8$};
	
	\draw (v11) -- (v21);
	\draw (v11) -- (v22);
	\draw (v11) -- (v23);
	\draw (v12) -- (v23);
	\draw (v12) -- (v21);
	\draw (v12) -- (v22);
	\draw (v12) -- (v31);
	
	\draw (v11) -- (v31);
	\draw (v11) -- (v32);
	\draw (v11) -- (v33);
	\draw (v12) -- (v32);
	\draw (v12) -- (v33);
	
	\draw (v21) -- (v32);
	\draw (v21) -- (v31);
	\draw (v21) -- (v33);
	\draw (v22) -- (v32);
	\draw (v22) -- (v31);
	\draw (v22) -- (v33);
	\draw (v23) -- (v31);
	\draw (v23) -- (v32);
	\draw (v23) -- (v33);
	
\end{tikzpicture}
\end{center}

Note that the following relation holds
\begin{equation}\label{eq:gr}
	K_{p,q,r} = K_n - (K_p \oplus K_q \oplus K_r),
\end{equation}
where the operation $-$ yields a graph obtained from $K_n$ after erasing the edges of $K_p \oplus K_q \oplus K_r$ appearing in $K_n.$
On level of the Hamiltonians, \eqref{eq:gr} implies that
$$
H^d_{K_{p,q,r}} = H^d_{K_n} -\big( H^d_{K_p} \otimes  I_{d^{q+r}} + 
I_{d^{p}} \otimes H^d_{K_q} \otimes I_{d^{r}}  + I_{d^{p+q}}\otimes H^d_{K_r}  \big)
$$
and thus for any irrep of $S_n$ indexed by a partition $\lambda \vdash n,$
\begin{equation}\label{eq:ham}
H^\lambda_{K_{p,q,r}} = H^\lambda_{K_n} -\big( (H^d_{K_p} \otimes  I_{d^{q+r}})^\lambda + 
(I_{d^{p}} \otimes H^d_{K_q} \otimes I_{d^{r}})^\lambda  + (I_{d^{p+q}}\otimes H^d_{K_r})^\lambda  \big).
\end{equation}

For a clique $K_n$ and any  partition $\lambda \vdash n,$ it is known that $H_{K_{n}}^\lambda = \eta_\la I$ is a scalar matrix (see \cite[Lemma 2.11]{BCEHK24} for the case $d=2;$ the proof extends directly to the case of a general $d$ \cite{KSV+}). 
Moreover, a precise formula for the eigenvalue $\eta_\lambda$ of $H_{K_{n}}^\lambda$ is known: 
\begin{proposition}{\cite[Proposition 6.7]{KSV+}}
	For any $\lambda \vdash n$ with rows $\lambda_1\ge\cdots\ge \lambda_d$,
	\begin{equation*}\label{eq:etaD}
		\eta_\lambda=
		n^2 +\frac{d(d-1)(2d-1)}{6}
		-\sum_{k=1}^d\big( \lambda_k - (k-1)\big)^2.
	\end{equation*}
\end{proposition}

Note that the matrices $H^d_{K_p} \otimes  I_{d^{q+r}}, I_{d^{p}} \otimes H^d_{K_q} \otimes I_{d^{r}}$ and $I_{d^{p+q}}\otimes H^d_{K_r}$ are in the 
 image of 
   the representation $\rho_n^{(d)}$ of $S_n$ and are hence contained in the subalgebra $\CC[S_{p}]\otimes \CC[S_q] \otimes \CC[S_r] \cong \CC[S_p\times S_q \times S_r]$ of $\CC[S_n]$.
To compute the eigenvalues of \eqref{eq:ham} we therefore need to consider the restriction of the irrep $\lambda$ of $S_n$ to the subgroup $S_p\times S_q \times S_r \subseteq S_n$. In fact, we need to understand how this restriction  decomposes as a direct sum of irreducible representations of $S_p\times S_q \times S_r$.
For this end, a branching rule is invoked, namely the Littlewood-Richardson rule \cite[Section 4.9]{Sag01} together with the Frobenius reciprocity \cite[Theorem 1.12.6]{Sag01}. More precisely, if $V_\la$ is the irreducible  module of $S_n$ corresponding to $\lambda\vdash n,$ then $V_\la$ decomposes as an $S_p\times S_q \times S_r$ module as
\begin{equation}\label{e:LR}
	V_\lambda^{\downarrow S_p\times S_q \times S_r}=\bigoplus_{\substack{\mu\vdash p,\  \nu\vdash q,\\ \zeta \vdash r}} 
	\big(V_\mu \otimes V_\nu \otimes V_\zeta\big)^{c^\lambda_{\mu\nu\zeta}},
\end{equation}
where $c^\lambda_{\mu\nu \zeta}$ is the (iterated) Littlewood–Richardson coefficient (see \cite[Section 4.9]{Sag01} and \cite{KLMS12,GL20}).

The matrices $H^d_{K_p} \otimes  I_{d^{q+r}}, I_{d^{p}} \otimes H^d_{K_q} \otimes I_{d^{r}}$ and $I_{d^{p+q}}\otimes H^d_{K_r}$  commute, so the eigenvalues of 
$$
(H^d_{K_p} \otimes  I_{d^{q+r}})^\lambda + 
(I_{d^{p}} \otimes H^d_{K_q} \otimes I_{d^{r}})^\lambda  + (I_{d^{p+q}}\otimes H^d_{K_r})^\lambda 
$$
are sums of matching eigenvalues of these three matrices. We have thus established the following proposition.

\begin{proposition}
	The solution of the $d$-QMC problem for $K_{p,q,r}$ is obtained by maximizing 
	\begin{equation}\label{eq:lmn}
		\Xi(\lambda, \mu, \nu, \zeta) = \eta_\lambda - \eta_\mu -\eta_\nu -\eta_\zeta,
	\end{equation}
	where $\lambda \vdash n = p+q+r, \mu \vdash p, \nu \vdash q, \zeta \vdash r$ are such that $c^\lambda_{\mu\nu\zeta}\neq 0.$
\end{proposition}

\subsection{Littlewood-Richardson coefficients}
Assume $(\lambda,\mu,\nu, \zeta)$ is a tuple of partitions such that $\la \vdash n, \mu \vdash p,\nu \vdash q, \zeta \vdash r$ with $p\geq q\geq r >0$ and $p+q+r=n.$ 
Assume $\mu,\nu,\zeta$ are all subpartitions of $\la$ in the sense that $\mu_i,\nu_i,\zeta_i \leq \la_i$ for all feasible $i.$

For a partition $\sigma \in \{\nu,\zeta\},$ we denote by $\lambda[\sigma]$ the skew-shaped subdiagram of $\lambda$ pertaining to $\sigma.$
A filling $T$ of the skew Young diagram $\lambda[\sigma],$ where $\sigma \in \{\nu,\zeta\},$ with natural numbers is a \textbf{Littlewood-Richardson (LR) tableau} if
\begin{enumerate}
	\item it is a \textit{semistandard Young tableau} (its entries weakly increase along each row and strictly increase down each column), and
	\item the concatenation of reversed rows in $T$ is a \textit{lattice word} (a word in which every prefix contains at least as many $i$s as $(i + 1)$s).
\end{enumerate}
The \textit{enumeration} of a (skew-shaped) tableau $T$ is the partition whose $i$-th part counts the number of $i$s in $T$. 
By \cite[Theorem 4.9.4]{Sag01} (see also \cite{KLMS12,GL20}), the iterated Littlewood-Richardson coefficient $c^\lambda_{\mu\nu\zeta}$ counts the number of
enumerations of two disjoint skew-shaped diagrams $\lambda[\nu]$ and  $\lambda[\zeta]$ (with $q$ and $r$ boxes respectively) whose union is $\la/\mu$  such that $\la / \lambda[\zeta]$ is a valid partition and both  $\lambda[\nu]$ and  $\lambda[\zeta]$ are Littlewood-Richardson tableaux.

\begin{definition}
	We will often refer to a tuple of partitions $(\la,\mu,\nu,\zeta)$ as \textbf{valid} if $\mu,\nu$ and $\zeta$ are subpartitions of $\la$ and the above conditions of enumerating the Young diagram of $\la$ are satisfied.
	For $\sigma \in \{\mu,\nu,\zeta\}$ we say that $\lambda[\sigma]$ is a \textbf{minimal LR tableau} if it is a LR tableau, where (going left to right, top to bottom) each box is filled with the smallest possible number. 
\end{definition}

\begin{example}
	Let $p=q=3,r=2$ and $\la=(3,3,2) \vdash 8, \mu = (2,1)\vdash 3, \nu = (2,1) \vdash 3.$ If we choose $\zeta= (2)$ or $\zeta = (1,1),$ then we obtain a valid tuple of partitions. In fact, the left-most figure below shows a minimal LR tableau (when $\zeta = (2)$) and the middle figure shows a valid but not minimal enumeration (when $\zeta = (1,1)$). However, keeping the colors of the boxes from the first two figures and
	choosing $\nu = (3)$ instead of $\nu = (2,1)$ (and, say, $\zeta = (2)$) yields a situation where no valid enumeration is attainable (see the right-most figure, where the columns in the gray block are not strictly increasing).
	
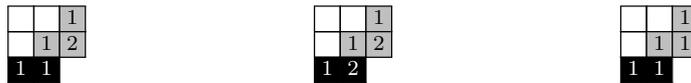
\begin{align*}
	\begin{ytableau}
		& 	&*(lightgray) 1\\
		&*(lightgray)1	& *(lightgray)2\\
		*(black) \textcolor{white}{1}& *(black)\textcolor{white}{1}\\
	\end{ytableau}
\hspace{3cm}
\begin{ytableau}
	& 	&*(lightgray) 1\\
	&*(lightgray)1	& *(lightgray)2\\
	*(black) \textcolor{white}{1}& *(black)\textcolor{white}{2}\\
\end{ytableau}
\hspace{3cm}
\begin{ytableau}
	& 	&*(lightgray) 1\\
	&*(lightgray)1	& *(lightgray)1\\
	*(black)\textcolor{white}{1} & *(black)\textcolor{white}{1}\\
\end{ytableau}
\end{align*} 
\end{example}

\begin{remark}
	As per the definition of the LR coefficient $c^\lambda_{\mu,\nu,\zeta}$, the Young diagram of the partition $\lambda \vdash n$ gets divided into 3 disjoint skew-shaped subdiagrams pertaining to partitions $\mu,\nu$ and $\zeta$ respectively. A necessary condition for $c^\lambda_{\mu,\nu,\zeta}$ to be nonzero is that $\mu \vdash p,\nu \vdash q, \zeta \vdash r$ are subpartitions of $\lambda.$
	
	 We will often represent the boxes pertaining to the 3 partitions with a distinguished color:  white for the $\mu$ boxes, grey for the $\nu$ boxes and black for the $\zeta$ boxes. By definition, the Young diagram of $\la$ is constructed so that first the white boxes are added (they should form a valid partition $\mu$), then the grey boxes are added (so that the union of all boxes is again a valid partition) and finally, the black boxes are added to complete the shape of $\la.$
\end{remark}

\subsection{Contents formula and Robin Hood moves}
To analyze $\Xi$ from \eqref{eq:lmn}, we use its explicit form known as the contents formula. The proof of the following proposition follows along the lines of \cite[Lemma B.1 and Lemma 6.5]{KSV+}. We include it for completeness.
\begin{proposition} Let $n, p,q,r \in \N$ be such that $p\geq q\geq r >0$ and $p+q+r=n.$  Let $(\la,\mu,\nu,\zeta)$ be a tuple of partitions with $\la \vdash n, \mu \vdash p, \nu \vdash q, \zeta \vdash r$ such that $\mu,\nu,\zeta$ are all subpartitions of $\la.$ Then
\begin{align}\label{eq:Xi}
	\begin{split}
		\Xi(\lambda,\mu,\nu, \zeta) =& 2 \big(n (p + q)-p^2 - p q - q^2\big) - 2  \sum_{c \in \lambda} \textnormal{content}(c) + \sum_{c \in \mu} \textnormal{content}(c)\\  +& \sum_{c \in \nu} \textnormal{content}(c) + \sum_{c \in \zeta} \textnormal{content}(c),
	\end{split}
\end{align}
where
the content of a box located in row $r$ and column $c$ is defined as $\text{content}(\text{box}) = c - r.$
\end{proposition}

\begin{proof}
	For any partition $\sigma$ let $\Sigma(\sigma) = \sum_{c \in \sigma}\text{content}(c).$ By \cite[Theorem 4]{Lassalle} or \cite{Fro01},  $$\Sigma(\zeta) = \binom{n}{2}\frac{\chi_\sigma((ij))}{\chi_\sigma(e)},$$
	where $\chi_\sigma$ is the character pertaining to the irrep of $S_n$ given by $\sigma.$ By \cite[Lemma 6.5]{KSV+},
	\begin{equation} \label{eq:eta_simplified}
		\eta(\sigma) = m^2 - m - 2 \Sigma(\sigma).
	\end{equation}
	Replacing $\eta_\lambda,\eta_\mu,\eta_\nu,\eta_\zeta$ in $\Xi(\la,\mu,\nu,\zeta)$ with \eqref{eq:eta_simplified} completes the proof of the proposition.
\end{proof}

Hence, to increase the value of $\Xi,$ one should either decrease the sum of contents of $\la$ or increase the sum of contents of $\mu,\nu$ or $\zeta.$ In other words, at least one $\la$ box should be moved down or left  or at least one $\mu,\nu$ or $\zeta$ box should be moved up or right.
Throughout, we shall assume that the appearing tableaux are minimal LR tableaux.
In this way, the enumerations of $\lambda[\nu]$ and $\lambda[\zeta]$
are also minimal in that the sum of contents of $\nu$ and $\zeta$ are minimal. Since we are maximizing $\Xi,$ the latter can be assumed w.l.o.g.~by \eqref{eq:Xi}.

\begin{remark}
	Note that if $d=1,$ then clearly, the solution to the $1$-QMC problem for $K_{p,q,r}$ is $0,$ attained at $\la =(n), \mu = (p), \nu = (q),\zeta = (r).$ The largest eigenvalue of the Hamiltonian can be easily computed by noting that the sum of contents of a partition $\sigma = (m) \vdash m$ equals $\binom{m}{2}.$
\end{remark}

We are now ready to properly state the main result (cf.~Theorem \ref{th:prelim}).

\begin{theorem}\label{th:main}
	Let $n, p,q,r \in \N$ be such that $p\geq q\geq r >0$ and $p+q+r=n.$
	\begin{enumerate}[(a)]
		\item The solution to the $2$-QMC problem for $K_{p,q,r}$ is attained at the tuple $(\la,\mu,\nu,\zeta)$ with $\la \vdash n, \mu \vdash p, \nu \vdash q, \zeta \vdash r,$ where $\mu=(p),\nu=(q),\zeta=(r)$ and
		\begin{enumerate}[(i)]
			\item    $\la= (p,q+r) = (p,n-p)$  if $p\geq q+r;$
			\item $\la$ is the balanced partition of $n$ of height $2,$ i.e., $\la_1-\la_2 \leq1,$ if $p < q+r.$
		\end{enumerate}
		In (i), the solution is $2(n-p)(p+1)$ and in (ii), the solution is $4k^2-1$ if $n=2k$ and $4k(k+1)-3$ if $n=2k+1.$
		
		\item The solution to the $3$-QMC problem for $K_{p,q,r}$ is  $2 n (2 + p + q) - 2 (p^2 + q + q^2 + p (2 + q)),$ attained at the tuple $(\la,\mu,\nu,\zeta)$ with $\la \vdash n, \mu \vdash p, \nu \vdash q, \zeta \vdash r,$ where $\la=(p,q,r), \mu = (p), \nu = (q)$ and $\zeta = (r).$
	\end{enumerate}
\end{theorem}
The proof of Theorem \ref{th:main} is divided into two parts: Subsection \ref{sec:d2} handles the case $d=2$ and Subsection \ref{sec:d3} addresses the case $d=3.$

\section{Proof of Theorem \ref{th:main}}
To solve the $d$-QMC problem for complete tripartite graphs when $d=2,3$ 
we begin with a lemma establishing that only partitions of maximal height can maximize the value of $\Xi$ as in  \eqref{eq:lmn}.

\begin{lemma}\label{le:maxd}  Let $d \in \{2,3\}$ and let $n, p,q,r \in \N$ be such that $p\geq q\geq r >0$ and $p+q+r=n.$  Let $(\la,\mu,\nu,\zeta)$ be a tuple of partitions with $\la \vdash n, \mu \vdash p, \nu \vdash q, \zeta \vdash r$ such that $\mu,\nu,\zeta$ are all subpartitions of $\la.$ Then the value of \eqref{eq:lmn} is maximal when $\la$ has height $d.$
\end{lemma}

\begin{proof}
	Assume that the height of $\la$ is $0<d'<d.$ Color the boxes pertaining to $\mu,\nu$ and $\zeta$ in white, grey and black, respectively.
	
	If $d'=1,$ we have $\mu=(p),\nu=(q),\zeta =(r).$ Now moving one black box to a new (i.e., the second) row does not change $\mu,\nu,\zeta$ and it decreases the sum of contents of $\la.$ Such a move thus increases the value of $\Xi.$ 
	We can now assume $d'=2, d=3.$ 
	Consider the last box in the last row, denoted by $\star.$
	We separate several cases:
	
	(a) If box $\star$ is white, then it is immediate that both the box right above it and the first box in the second row are also white. Note that in this case $\zeta = (r).$  Now move one black box (from the first row) to a new row (i.e., to the third row). Doing so, $\mu, \nu$ and $\zeta$ stay the same, however, the sum of contents of $\la$ decreases (since we move the box down and left). Hence, the value of $\Xi$ increases.
	
	(b) If box $\star$ is black, then consider the box right above it, denoted by $\star',$ and the first box of the second row, denoted by $\star''.$
	Note that  $\star'$ and $\star''$ cannot be both black.
	
	If $\star'$ and $\star''$ are both white or grey, then move $\star$ to a new row.
	In this case $\mu,\nu$ and $\zeta$ do not change after the move and the sum of contents of $\la$ decreases, thus the value of $\Xi$ increases.
	
	If $\star'$ is black and $\star''$ is either white or grey, then again move $\star$ to a new row. In this case, to maintain a minimal enumeration of the black boxes, one box from the second row of $\zeta$ gets moved to the first row. This increases the sum of contents of $\zeta.$ The partitions $\mu$ and $\nu$ do not change after the move and the sum of contents of $\la$ decreases (as the box $\star$ is moved down and left). Hence, the value of $\Xi$ increases.
	
	It remains to consider the case when $\star'$ is either white or grey and $\star''$ is black.	If there is at least one black box in the first row, then move it to a new (i.e., the third) row. Doing so, the sum of contents of $\zeta$ decreases, but  the sum of contents of $\la$ decreases even more. Indeed, note that since $p\geq q \geq r,$ we have $\la_1 \geq 2 \la_2 + 1$ and $\zeta_1 < \la_1.$ Hence the decrease of the content of the box moved in $\zeta$ is lower than the decrease of the content of the black box moved in $\la.$ The value of $\Xi$ thus increases after the move.
	
	 If there are no black boxes in the first row, then we must have $\la_1=p+q, \mu=(p),\nu=(q),\la_2 = r,\zeta = (r).$ In this case move all the black boxes to the third row  and move all grey boxes to the second row (a series of so-called Robin Hood moves). Doing so, $\mu, \nu$ and $\zeta$ stay the same and sum of contents of $\la$ decreases, whence the value of $\Xi$ increases.
	 
	 (c) If box $\star$ is grey,  again consider the boxes $\star'$ and $\star''$ as defined in (b) above.  Note that in this case $\star'$ and $\star''$ can only be white or grey.
	 
	 If $\star'$ and $\star''$ are both white or both grey or if $\star'$ is grey and $\star''$ is white, then similar arguments as in (b) show that moving $\star$ to a new row increases the value of $\Xi.$
	 
	 If $\star'$ is white and $\star''$ is grey, then all boxes in the second row are grey. This means that the last box in the first row is black (and $\zeta = (r)$). Now moving one black box (from the first row) to a new row does not change $\mu,\nu$ and $\zeta$ and it clearly decreases the sum of contents of $\la.$ Hence, the value of $\Xi$ increases.
\end{proof}

\subsection{The case $d=2$} \label{sec:d2}


	Recall the convention from Lemma \ref{le:maxd} of coloring the boxes pertaining to $\mu,\nu$ and $\zeta$ in white, grey and black, respectively.
	By Lemma \ref{le:maxd}, to maximize the value of $\Xi,$ the height of $\la$ must be $2.$ Let $(\la,\mu,\nu,\zeta)$ be any tuple of partitions with $\la \vdash n, \mu \vdash p, \nu \vdash q, \zeta \vdash r$ and $p+q+r=n$ such that $\mu,\nu,\zeta$ are all subpartitions of $\la.$

%
%
	
(a) We first prove that moving all black boxes to the second row increases $\Xi.$
So assume there is at least one black box in the first row.
\begin{enumerate}[(i)]
	\item If there are only black boxes in the second row, then all black boxes in the first row can be moved to the second row to form a valid partition (since $p\geq q \geq r$). In this way $\mu,\nu$ and $\zeta$ do not change (they are all one-row partitions) and the sum of contents of $\la$ decreases (since boxes are moved down and left). Hence, $\Xi$ increases.
	
	\item If there are white or grey boxes in the second row, then exchange the right-most non-black box  in the second row (denote it $\dagger'$) with the left-most black box in the first row (denote it $\dagger$). Note that box $\dagger$ is to the left of box $\dagger'$ and has a white or grey box right above it. 
	In any case the sum of contents of $\nu$ increases after the move: if $\dagger'$ had a white box right above it, $\nu$ stays the same, and if $\dagger$ had a grey box right above it, one box from the second row of $\nu$ moves to the first row of $\nu.$
	On the other hand, box $\dagger$ has either another black box or nothing right below it. In any case the sum of contents of $\zeta$ increases after the move: if $\dagger$ had nothing right below it, $\zeta$ stays the same, and if $\dagger$ had a black box right below it, one box from the second row of $\zeta$ moves to the first row of $\zeta.$
	Since $\mu$ and $\la$ do not change, the described move increases the value of $\Xi.$
	
	Now repeat the same move as long as possible. Such a move is no longer possible when either all the black boxes are in the second row, in which case we are done, or there are still black boxes in the first row, but no non-black boxes left in the second row, in which case we are in the situation of (i).
	
\end{enumerate}

(b)	Now assume all black boxes are in the second row (note that any grey or white boxes might still be in the second row as well). We shall prove that moving all white boxes to the first row increases the value of $\Xi.$

\begin{enumerate}[(i)]
	\item If there are any white boxes in the second row and at least one grey box in the first row, use the same procedure as in (ii) from part (a) above to exchange the grey boxes in the first row with the white boxes in the second row. When this construction no longer applies, there are either no grey boxes left in the first row or there are no white boxes left in the second row. 
	\item If the first row only has white boxes and there is at least one white box in the second row, we must have $p\geq q+r.$ In this case move the white boxes from the second to the first row. Doing so, the sum of contents of $\la$ increases by $a(\la_1-\la_2+2),$ where $a$ is the number of white boxes in the second row, but the sum of contents of $\mu$ increases by $a(\la_1-a + 1),$ which is bigger than  $a(\la_1-\la_2+2)$ since $\la_2 \geq a+1.$
	Since $\nu$ and $\zeta$ do not change, the value of $\Xi$ increases.
	Now $\mu,\nu,\zeta$ are optimal to maximize $\Xi$ and reversing the above argument shows that changing $\la$ to make its two part closer in length (which would decrease the sum of contents of $\la$) decreases the value of $\Xi.$ So we have found the tuple of partitions giving the maximal value of $\Xi.$
	\item If the first row has at least one grey box and the second row has no white box (so all the white boxes are in the first row), then we must have $\mu = (p),\nu=(q)$ and $\zeta = (r)$ at this point. So $\mu,\nu,\zeta$ are optimal to maximize $\Xi.$
	
	If $p\geq q+r,$ then move all the grey boxes to the beginning of the second row. This moves do not change $\mu,\nu,\zeta$ and they decrease the sum of contents of $\la.$ Thus, the value of $\Xi$ increases. Now changing $\la$ to make its two parts closer in size decreases the value of $\Xi$ as seen in (ii) above (from part (b)). We have hence found the tuple of partitions giving the maximal value of $\Xi.$
	The same can be deduced if the first row only has white boxes and the second row has no white boxes.
	
	If $p< q+r,$ then move some of the grey boxes from the first row to the second row to make $\la$ balanced. Now all four partitions $(\la,\mu,\nu,\zeta)$ are optimal to maximize $\Xi$ and we are done.

\end{enumerate}
 
 We have thus shown that $\Xi$ is maximized at the tuple $(\la,\mu,\nu,\zeta)$ with $\la \vdash n, \mu \vdash p, \nu \vdash q, \zeta \vdash r,$ where $\mu=(p),\nu=(q), \zeta=(r)$  and either $\la= (p,q+r)$ if $p\geq q+r$ or  
  $\la$ is the balanced partition of $n$ of height $2$ if $p < q+r.$ 
  The maximal eigenvalue of the $2$-QMC  Hamiltonian in the two cases $p\geq q+r$ and $p<q+r$ can be now directly computed using \eqref{eq:Xi}.\looseness=-1

\subsection{The case $d=3$}\label{sec:d3}


 Recall the convention of coloring the boxes pertaining to $\mu,\nu$ and $\zeta$ in white, grey and black, respectively.
	By Lemma \ref{le:maxd}, to maximize the value of $\Xi,$ the height of $\la$ must be $3.$

	(a) The first goal is to move any white boxes from the third row to the first row.
	\begin{enumerate}[(i)]
		\item 	If there are only white boxes in the first row, move all the white boxes from the third row to the first row and shift the remaining boxes in the third row to the left. Note that, before the move, the boxes right above the white ones in the third row were also white, so the obtained tuple of partitions is still valid.
		After this move, the sum of contents of $\la$ increases (some of the boxes in its third row are moved to its first row, this lowers the value of $\Xi$), but the sum of contents of $\mu$ increases at least as much. This is because when computing the change of contents of the $\la$ boxes, we consider them moving from the end of the third row (i.e., $\la$ does not distinguish between different shades of boxes), whereas the $\mu$ boxes are moved from the beginning of the  third row.
		
		The sum of contents of $\nu$ ($\zeta$ resp.) also increases, since the number of shared columns between the grey (black resp.) boxes in the second and third row decreases. Indeed, $\nu$ ($\zeta$ resp.) has at most $2$ rows at start and after the move $\nu$ ($\zeta$ resp.) either stays the same or some of its boxes in the second row move to the first row. 
		We conclude that the value of $\Xi$ does dot decrease after the described move.
		
		\item 
		If there are grey or black boxes in the first row, then exchange the left-most non-white box in the first row (denoted by $\star$) with the right-most white box in the third row (denoted by $\star'$). Note that $\star$ is strictly to the right of $\star'$ and the box right above $\star'$ is white.  If $\star$ is black and the box to the right of $\star'$ is grey, then after the exchange, shift $\star$ (now positioned in the third row) at the end of the grey boxes in the third row to make the new tuple of partitions $(\la,\mu,\nu,\zeta)$ valid.

		After this move, $\la$ does not change and the sum of contents of $\mu$ increases (one box gets moved up and right). The sum of contents of $\nu$ either stays the same or increases. Indeed,  if $\star$ is grey, then $\nu$ either stays the same (if the box right below $\star$ is black) or its sum of contents increases (if the box right below $\star$ is grey) since one of its boxes is move up and right. Moreover, if $\star$ is black, (this implies that $\nu$ has at most $2$ rows at start and) the grey boxes in the third row get shifted to the right after the move, which means that some of the boxes in the second row of $\nu$ might move to the first row. If the latter happens, the sum of contents of $\nu$ increases.
		
		Similarly, the sum of contents of $\zeta$ either stays the same or increases after the move. Indeed, note that the boxes right after the grey boxes in the third row are either white or grey. Hence, if $\star$ is black, it is labelled by $1$ (i.e., lies in the first row of $\zeta$) before and after the move. If the box right below $\star$ is grey at start, $\zeta$ does not change. However, if the box right below $\star$ is black at start, then the sum of contents of $\zeta$ increases as one of its boxes is moved up and right after the move. It is clear that $\zeta$ stays the same after the move if $\star$ is grey.
		Hence, one such exchange of boxes increases the value of $\Xi.$
	\end{enumerate}
	
		Repeat this step several times until there are either no white boxes in the third row, in which case we move to (b), or there are no more non-white boxes in the first row, in which case we use construction from (i) above.

\vspace{5mm}
	(b) When there are no white boxes in the third row, we proceed by moving all black boxes from the first row to the second or third row.
	
	Denote the left-most black box in the first row by $\star$ and denote the right-most non-black box in the third row by $\star'$ (note that $\star'$ is in fact grey). If there is any box right below $\star,$ it is a black one. On the other hand, the box right above $\star'$ is white or grey. 
	In any case both $\star$ and $\star'$ can be labelled by $1$ after the exchange. Hence, the sum of contents of $\nu$ and $\zeta$ either stay the same or increase after the move. Since $\la$ and $\mu$ do not change, the value of $\Xi$ does not decrease after this exchange.

	Repeat this exchange until there are either no black boxes in the first row, in which case we proceed with (c) below, or there are no grey boxes in the third row.
	If the latter happens, then since $p \geq q \geq r,$ there is a place either in the third row or in the second row, where a black box can be put so that right above it there is a non-black box.  This move does not increases the value of $\Xi$ as $\mu,\nu,\zeta$ stay the same and the sum of contents of $\la$ increases as one of its boxes gets moved down and left.
	Repeat this process until there are no black boxes in the first row.
	\vspace{5mm}

	(c) When there are no white boxes in the third row and no black boxes in the first row, we continue by moving all white boxes from the second row to the first row.
	
	\begin{enumerate}[(i)]
		\item Exchange the left-most grey box in the first row (denoted by $\star$) with the right-most white box in the second row (denoted by $\star'$). 
		The sum of contents of $\mu$ clearly increases after this move and $\la$ and $\zeta$ do not change.
		
		The box right under $\star'$ before the move was either grey or black.
		 If it was black, then no matter what was right below $\star$ (grey, black or nothing), the sum of contents of $\nu$ increases after the move. Indeed, after the move, $\nu$ either stays the same (if the box under $\star$ was black or if there was an empty space) or one of its boxes gets moved to a higher row (if the box under $\star$ was grey).
		 
		 Assume the box right under $\star'$ before the move was grey. If the box right under $\star$ was grey as well, then after the move, $\nu$ either does not change (the box in the third row two places under $\star$ was black or if there was an empty space) or one of its boxes gets moved to a higher row, so up and right (if the box in the third row two places under $\star$ was grey).
		 If the box right under $\star$ was black or if there was an empty space, then it can happen that the decrease in the sum of contents of $\nu$ is greater than the increase in the sum of contents of $\mu,$ so the move must be modified as follows.
		 
		 Assume first there is at least one black box in the second row. Denote the left-most black box in the second row by $\star''.$ If the box right above $\star''$ is white, then first exchange $\star''$ with the grey box right below $\star'$ (denote this grey box by $\star'''$). If this move produces a feasible tuple of partitions, then it clearly increases the value of $\Xi.$ After that, $\star$ and $\star'$ can be exchanged as explained above.
		 The move involving $\star''$ is valid if the box  to the right of $\star'''$ is black. If this is not the case, then exchange $\star''$ with the right-most grey box in the third row. This move does not decrease the value of $\Xi$ as it is easy to see that $\nu,\zeta$ (as well as $\la,\mu$) stay the same after this move. Repeat exchanging black boxes in the second row with grey boxes in the third row until there are either no black boxes in the second row (in which case proceed with (ii) below) or the left-most black box in the third row is the right neighbour of $\star'''.$ If the latter is true, then exchanging $\star$ and $\star'$ betters the value of $\Xi$ as already explained.

		 If the box right above $\star''$ is grey instead of white and if the box to the right of $\star'''$ is black, then by assumption we have that $\star''$ is the box right under $\star.$ Now permute the four starred boxed as follows: $\star \to \star'' \to \star''' \to \star' \to \star.$ After this move $\la,\nu, \zeta$ do not change and the sum of contents of $\mu$ increases. Hence, the value of $\Xi$ increases. 
		 If the box to the right of $\star'''$ is grey, then exchange $\star''$ with the grey box right next to $\star'''.$ This move does not decrease the value of $\Xi$ since $\nu$ does not change (the box to the right of $\star'''$ was labelled by $2$ and it keeps the label after the move) and the sum of contents of $\zeta$ either stays the same or increases (it only increases if the box below $\star''$ was black). Continue exchanging black boxes in the second row with grey boxes in the third row until there are either no black boxes in the second row  (in which case move to (ii) below) or the left-most black box is the right neighbour of $\star'''.$
		  Now exchanging $\star$ and $\star'$ betters the value of $\Xi$ as explained above (the boxes under $\star$ and $\star'$ are both grey).

		\item  Now if there are no black boxes in the second row, then all black boxes are in the third row. As before, denote the left-most grey box in the first row by $\star$ and the right-most white box in the first row by $\star'.$ We assume the problematic situation where $\star$ and $\star'$ cannot be exchanged to better the value of $\Xi,$ i.e., we assume that there is an empty space below $\star$ and there is a grey box under $\star'.$
		Denote the box right under $\star'$ by $\star''.$ Now exchange $\star$ and $\star'$ and move $\star''$ from the third row to the second row. Note that since $\star$ had no box right under itself, we deduce that after the move, $\star''$ has a white box right above itself. Doing so, $\nu$ and $\zeta$ do not change and the sum of contents of $\la$ increases since one of its boxes is moved from the third to the second row. However, it is easy to see that the sum of contents of $\mu$ increases even more (one of its boxes is moved from the second to the first row) since the the $\mu$ box that gets moved starts to the left of the $\la$ box that gets moved. Hence, the value of $\Xi$ increases.

		Repeat this process of exchanging the grey boxes in the first row with the white boxes in the second row until there are either no grey boxes in the first row, in which case proceed with (iii) below,
		or there are no white boxes in the second row, in which case proceed with (d) below.

%
%
		
		\item When there are only white boxes in the first row and at least one white box in the second row, the goal is to move all grey boxes to the second row.
		 Begin by exchanging the left-most black box in the second row with the right-most grey box in the third row. The black box is labelled by $1$ before and after the move, and the grey box is labelled by $1$ or $2$ before the move and it is labelled by $1$ at the end. Hence, this exchange increases the value of $\Xi.$ 
		 
		 Repeat this process until there are no black boxes in the second row (which we want) or there are no grey boxes in the third row, but there are still black boxes in the second row. In the latter case, we move the remaining black boxes from the second row to the third row. This results in a valid partition since $p\geq q \geq r$ and clearly increases the value of $\Xi$ (since $\zeta$ has the optimal 1-row shape at the end and the sum of contents of $\la$ increases as we move boxes down and left).
		 
		 We can now assume the case where the first row only has white boxes, the third row only has black boxes and the second row has both white and grey boxes. Now moving the white boxes to the first row increases the value of $\Xi.$  This is because when computing the change of contents of the $\la$ boxes, we consider them moving from the end of the second row (i.e., $\la$ does not distinguish between different shades of boxes), whereas the $\mu$ boxes are moved from the beginning of the  second row.
	\end{enumerate}

(d)  When all white boxes are in the first row and there are no black boxes in the first row, we continue by moving all black boxes from the second row to the third row.

Exchange the left-most black box in the second row (denoted by $\star$) with the right-most grey box in the third row (denoted by $\star'$). We know that the box right above $\star'$ is grey and the box right above $\star$ is not black. So $\star'$ is labelled by $2$ before the move and after, it is labelled by either $1$ or $2.$ The sum of contents of $\nu$ thus does not decrease. The same can be said for $\star:$ it is labelled by $1$ or $2$ before the move and it is labelled by $1$ after the move. Hence, the sum of contents of $\nu$ does not decrease. Since $\mu$ and $\la$ stay the same after the move, the value of $\Xi$ does not decrease.

Repeat this process until there are either no black boxes in the second row (in which case proceed to (e) below) or there are no grey boxes in the third row. In the latter case, the first row only has white and grey boxes, the second row has grey and black boxes and the third row only has black boxes. So move all grey boxes from the first to the second row and move all black boxes from the second to the third row. This clearly produces a valid partition $\la$ since $p \geq q \geq r,$ it decreases the sum of contents of $\la$ as its boxes are moved down and left and it does not decrease the sums of contents of $\nu$ and $\zeta$ (since they both have one row at the end). The value of $\Xi$ thus increases.


\vspace{5mm}

(e) When all white boxes are in the first row and all black boxes are in the third row, we move all grey boxes from the first and third row  to the second row.
\begin{enumerate}[(i)]	
	\item We first move all grey boxes from the first row to the second row. Doing so, $\mu,\nu$ and $\zeta$ do not change (note that since all white boxes are in the first row, the right-most grey box in the second row is to the left of the left-most grey box in the first row). Moreover, the sum of contents of $\la$ decreases. Hence, the value of $\Xi$ increases.

	\item Finally, we move all grey boxes from the third row to the second row. 
Doing so, $\mu$ and $\zeta$ do not change, the sum of contents of $\la$ increases (some of the $\la$ boxes get moved from the third to the second row), but the sum of contents of $\nu$ increases even more (the second row of $\nu$ gets joined with the first row), which implies that the value of $\Xi$ increases.
Indeed, as we have argued before, when computing the change of contents of the $\la$ boxes, they are thought of moving from the end of the third row to the second row, whereas the $\nu$ boxes are moved from the beginning of the  third row to the second row, causing a larger increase in the sum of contents.	
\end{enumerate}
We have thus shown that $\Xi$ is maximized at the tuple $(\la,\mu,\nu,\zeta)$ with $\la \vdash n, \mu \vdash p, \nu \vdash q, \zeta \vdash r,$ where $\la=(p,q,r), \mu = (p), \nu = (q)$ and $\zeta = (r).$
	The maximal eigenvalue of the $3$-QMC  Hamiltonian can be now directly computed using \eqref{eq:Xi}.

\end{document}